\documentclass{fundam}

\usepackage{thmtools,amsfonts,amsmath,amssymb}

\usepackage{url,amsmath,amssymb}
\usepackage{qtree}

\usepackage{complexity}
\usepackage{color}
\usepackage{caption}
\usepackage{verbatim,graphicx}

\usepackage[normalem]{ulem}		
\usepackage[ruled,vlined,linesnumbered]{algorithm2e}

\usepackage{tikz}
 \usetikzlibrary{trees}

\usepackage{forest}

\usepackage{mdframed}
\newcommand{\ourproblem}[1]{\begin{mdframed} #1 \end{mdframed}}

\newcommand{\rank}{{\mathrm{rank}}}

\newcommand{\HROTDIST}{\textrm{\sc RotDist}_H}
\newcommand{\RROTDIST}{\textrm{\sc RotDist}_R}
\usepackage{wrapfig}
\usepackage{enumitem}

\usepackage{hyperref}

\usepackage{todonotes}
\newcounter{todocounter}

\begin{document}


\setcounter{page}{79}
\publyear{24}
\papernumber{2173}
\volume{191}
\issue{2}

\finalVersionForARXIV

\title{On Rotation Distance of Rank Bounded Trees\footnote{A preliminary version of the
                   paper containing a subset of results appeared in the proceedings of the 28th International Computing and
                   Combinatorics Conference (COCOON 2022).}}

\author{Anoop S.K.M. \ Jayalal Sarma$^*$\thanks{Corresponding author: Department of Computer Science and Engineering,
                          Indian Institute of Technology Madras (IIT Madras), Chennai, India. \newline \newline
                    \vspace*{-6mm}{\scriptsize{Received April 2023; \ accepted March 2024.}}}
\\
Department of Computer Science and Engineering\\
Indian Institute of Technology Madras (IIT Madras)\\
Chennai, India\\
\{skmanoop | jayalal\}@cse.iitm.ac.in }


\runninghead{S.K.M. Anoop and J. Sarma}{On Rotation Distance of Rank Bounded Trees}

\maketitle

\begin{abstract}
Computing the rotation distance between two binary trees with $n$ internal nodes efficiently (in $\poly(n)$ time) is a long standing open question in the study of height balancing in tree data structures. In this paper, we initiate the study of this problem bounding the rank of the trees given at the input (defined in \cite{EH89} in the context of decision trees).

We define the \textit{rank-bounded rotation distance} between two given full binary trees $T_1$ and $T_2$ (with $n$ internal nodes) of rank at most $r = \max\{\rank(T_1),\rank(T_2)\}$, denoted by $d_R(T_1,T_2)$, as the length of the shortest sequence of rotations that transforms $T_1$ to $T_2$ with the restriction that the intermediate trees must be of rank at most $r$. We show that the rotation distance problem reduces in polynomial time to the rank bounded rotation distance problem. This motivates the study of the problem in the combinatorial and algorithmic frontiers. Observing that trees with rank $1$ coincide exactly with skew trees (full binary trees where every internal node has at least one leaf as a child), we show the following results in this frontier:
\begin{itemize}
\item We present an $O(n^2)$ time algorithm for computing $d_R(T_1,T_2)$. That is, when the given full binary trees are skew trees (we call this variant the \textit{skew rotation distance problem}) - where the intermediate trees are restricted to be skew as well. In particular, our techniques imply that for any two skew trees $d_R(T_1,T_2) \le n^2$.
\item We show the following upper bound: for any two full binary trees $T_1$ and $T_2$ of rank  $r_1$ and $r_2$ respectively, we have that: $d_R(T_1,T_2) \le n^2 (1+(2n+1)(r_1+r_2-2))$ where $r = \max\{r_1,r_2\}$. This bound is asymptotically tight for $r=1$.
\end{itemize}

En route to our proof of the above theorems, we associate full binary trees to permutations and relate the rotation operation on trees to  transpositions in the corresponding permutations. We give exact combinatorial characterizations of permutations that correspond to full binary trees and full skew binary trees under this association. We also precisely characterize the transpositions that correspond to the set of rotations in full binary trees. We also study bi-variate polynomials associated with binary trees (introduced by \cite{WG14}), and show characterizations and algorithms for computing rotation distances for the case of full skew trees using them.
\end{abstract}

\section{Introduction}
\label{sec:intro}

Rotation of full binary trees (every internal node has exactly two children, we refer to such trees simply as binary trees) is an important operation due to their classical applications in data structures.
A right~(resp. left) rotation of a full binary tree at an internal node consists of moving the right~(resp. left) subtree of the left~(resp. right) child of the node as the left~(resp. right) child of the node.
The \textit{rotation distance} between two full binary trees $d(T_1,T_2)$ is the minimum number of rotations required to transform $T_1$ to $T_2$.

\smallskip
There are two natural questions in this context. The first one is about proving (tight) upper and lower bounds for this distance measure - what can be the maximum distance between two trees having $n$ internal nodes. The second one is computational - given two  full binary trees, $T_1$ and $T_2$, with $n$ internal nodes, compute the rotation distance between them.
The decision version (we call it the {\sc RotDist} problem in the rest of the paper)
is: given $T_1, T_2$ and an integer $k \in \mathbb{N}$ check if $d(T_1,T_2) \le~k$.

\smallskip
Every full binary tree with $n$ internal nodes, by systematically applying at most $n-1$ right rotations at the internal nodes on the right path, can be converted to one where each internal node has an internal node as the right child and leaf as its left child.  Such a tree is called a \textit{right comb} (symmetrically the term \textit{left comb}). Since each such rotation can also be undone by symmetrically opposite left or right rotation, this gives that $d(T_1, T_2) \le 2n-2$ for any pair of trees $T_1$ and $T_2$ with $n$ internal nodes. The rotation distance between two trees is at most $2n-2$. For every $n$, does there exist two trees $T_1$, $T_2$ with $n$ internal nodes such that $d(T_1,T_2) = 2n-2$? This question has attracted a lot of attention in the literature~\cite{CW82}\cite{Deh09} and a lot of surprising connections were discovered while trying to answer the same~\cite{STT86}. Through a sequence of improvements, explicit trees were constructed~\cite{Deh09} for every $n \in \mathbb{N}$ such that $d(T_1,T_2) \ge 2n-O(\sqrt{n})$~\cite{Deh09}. \cite{STT86} initially conjectured the lower bound to be $2n-10$. The conjecture was later proved by~\cite{Pou14}.

\smallskip
The computational version of the rotation distance problem is still evasive despite a lot of attempts~(see \cite{CJ10} and references therein). In general, it is a long standing open question~\cite{STT86} whether the problem of computing $d(T_1,T_2)$ is $\NP$-hard or not. The best known approximation algorithm known for the optimization problem is a $2$-approximation algorithm~\cite{CJ10}. A fixed parameter tractable algorithm is also known~\cite{CJ10} for the problem when parameterized by rotation distance. Special cases are known to be linear time solvable when one of the trees is restricted to an angle tree (a generalization of the right and left combs)\cite{Luc04}.

\smallskip
A particularly important view of the computational problem is viewing it as a shortest path problem in an associated graph. For every $n \in \mathbb{N}$, consider the graph ${\cal R}_n$ with the vertex set as the set of all full binary trees with $n$ internal nodes and edges defined as follows: two vertices have an edge between them if the corresponding trees are at rotation distance $1$ apart. With this formulation, rotation distance is exactly the shortest path between two  given nodes (corresponding to the two given trees $T_1$ and $T_2$). The graph has exponential (in $n$) number of vertices (more precisely, the $n$-th Catalan number $C_n$) and each vertex has degree exactly $n-1$ since one rotation is associated with each internal node except the root node. The graph has logarithmic diameter (in terms of number of vertices, $C_n$) and the graph is also known to be hamiltonian~\cite{Lucas1987}. \\[-3mm]

\smallskip
\noindent{\bf Our Results:} We study the rotation distance problem by restricting a parameter of the tree known as the rank of full binary trees (introduced in the context of decision trees \cite{EH89}). The rank of a full binary tree is defined inductively as follows: the rank of a node at the leaf is $0$, and the rank of any node with two children is the maximum of the rank of the children nodes if their ranks are unequal, and as one plus the rank of the children if both children have equal rank. It can also be seen that rank is exactly the maximum height of the complete binary tree that can be obtained from the given tree by a sequence of edge contractions (and hence a minor of the given tree, since deletion of vertices will disconnect the tree). \cite{CG13} also provides a characterization of the rank of a binary tree using path decompositions.

\medskip
Towards our study, we formulate the following computational variant of the rotation distance problem (which we denote as the $\RROTDIST$ problem). Consider two full binary trees $T_1$ and $T_2$ and let $r=\max\{\rank(T_1),\rank(T_2)\}$. We define a path from $T_1$ to $T_2$ in the rotation graph (${\cal R}_n)$ to be a \textit{rank-bounded path}
if each intermediate tree that appears in the path has rank at most $r$. We define the \textit{rank bounded rotation distance}, $d_R(T_1,T_2)$, to be the length of the shortest rank-bounded path between the two trees $T_1$ and $T_2$ of rank at most $r$. The $\RROTDIST$ problem is defined as:

\ourproblem{$\RROTDIST$: Given two full binary trees $T_1$ and $T_2$, is the rank bounded rotation distance between the trees at most $k$ or not. That is, is $d_R(T_1,T_2) \le k$? }

We first show that the above variant is interesting since the general problem of computing rotation distance reduces to this problem - we record this as the following theorem.
\begin{restatable}{theorem}{rankboundedrotationdistancereductiontheorem}\label{thm:rankboundedrotationdistancereductiontheorem} {\sc RotDist} many-one reduces to $\RROTDIST$ in polynomial time.
\end{restatable}

This motivates a close study of rank-bounded paths and distance between two given trees. We also set out to understand the combinatorial upper bounds for rotation distance with respect to  this new measure. We first observe that rank 1 trees exactly correspond to skew trees (trees in which every internal node has a leaf as a child). This also shows that the trivial bound of $2n-2$ need not be rank bounded, since even the skew trees will become non-skew while being rotated to the right comb tree as indicated in the discussion above.

The problem of computing $d_R(T_1,T_2)$ - given two full skew trees $T_1$ and $T_2$, can we compute the rotation distance between the two trees where each individual intermediate tree in the path is also skew? We call this  the \textit{skew rotation distance problem}. We show a polynomial time algorithm for this problem.

\begin{restatable}{theorem}{skewtreesalgo}
\label{thmalgo:skewtrees}
Given two skew trees $T_1$ and $T_2$ on $n$ internal nodes, the skew rotation distance between $T_1$ and $T_2$ can be computed in time $O(n^2)$.
\end{restatable}

We show a distance upper bound for the distance in general.
\begin{restatable}{theorem}{rankdistbound}
\label{thm:rankdistbound}
For any two full binary trees on $n$ internal nodes $T_1$ and $T_2$ of rank $r_1$ and $r_2$ respectively, we have that $d_R(T_1,T_2) \le n^2 (1+(2n+1)(r_1+r_2-2))$ where $r = \max\{r_1,r_2\}$.
\end{restatable}
We define the height bounded rotation distance, $d_H(T_1,T_2)$, to be the length of the shortest height bounded path between the two trees $T_1$ and $T_2$ of height at most $h$. The $\HROTDIST$ problem is defined as:
\ourproblem{$\HROTDIST$: Given $h$, $k$ and two full binary trees $T_1$ and $T_2$ of height at most $h$, if the height bounded rotation distance between them is at most $k$ or not. That is, is $d_H(T_1,T_2) \le k$? }

We also show that the general problem of computing rotation distance reduces to this problem as well.

\begin{restatable}{theorem}{heightboundedrotationdistancereductiontheorem}\label{thm:heightoundedrotationdistancereductiontheorem} {\sc RotDist} many-one reduces to $\HROTDIST$ in polynomial time.
\end{restatable}

\noindent{\bf Our Techniques:}
We now discuss the techniques we use to prove Theorem~\ref{thmalgo:skewtrees} and Theorem~\ref{thm:rankdistbound} and related combinatorial bounds and characterizations. We believe that these might be of independent interest in other contexts too.

\medskip
We present a novel tool to analyse rotation distance problem by associating a permutation with every full binary tree as follows: label the internal nodes of the tree in such a way that the in-order traversal corresponds to the identity permutation of the integers $1,\ldots,n$. We call this labeling  the {\em in-order labeling}. The associated permutation to the full binary tree is the permutation produced by the pre-order traversal of the tree whose internal nodes are labeled with the in-order labeling. Since a binary tree is uniquely determined by the in-order and pre-order traversals, this is well-defined. We answer several combinatorial questions related to this new connection:\vspace*{-2mm}

\paragraph{\bf Characterizing Tree-permutations:}
Not all permutations in $S_n$ correspond to a tree, so we refer to those permutations which do correspond to a tree as {\em tree-permutations}. We now give an exact characterization of which permutations in $S_n$ are tree-permutations by using a pattern avoidance property.\vspace*{-2mm}
\begin{restatable}{theorem}{treepermutationcharacterization}\label{thm:treepermutation-charactersization}
A permutation $\sigma \in S_n$ is a tree permutation if and only if there does not exist indices $i<j<k$ such that $\sigma(k) < \sigma(i) < \sigma(j)$.
\end{restatable}

A full binary tree in which all caret nodes have at most one caret node on its left or right child is called a skew tree. We prove a characterization of tree permutations corresponding to skew trees (which we call  {\em skew permutations}) as follows:

\begin{restatable}{lemma}{introlemmaminmaxskew}
\label{minmaxSkew}
A permutation $\sigma \in S_n$ is a skew permutation if and only if $\forall i \in [n]$,
$$\sigma(i) = \min\{\sigma(i),\sigma(i+1),\ldots ,\sigma(n)\} \textrm{ ~~~or~~~ } \sigma(i) = \max\{\sigma(i),\sigma(i+1),\ldots ,\sigma(n)\}$$ where $\min$ (resp. $\max$) represent minimum (resp. maximum) from the given set.
\end{restatable}

Now the main idea of the algorithm for Theorem~\ref{thmalgo:skewtrees} is to translate the rotation problem for skew trees to an operation on binary strings which we prove in Lemma~\ref{minmaxswap} using the above characterization. En route to this proof, we also show that the skew rotation distance between two skew trees of $n$ internal nodes can be at most $n^2$ distance. We also analyse the distance of any given tree $T$ to the nearest skew tree and prove matching upper and lower bounds.\vspace*{-2mm}

\paragraph{\bf Characterizing Rotation using Transpositions:}
Given two trees $T_1$ and $T_2$, we can interpret them as permutations $\sigma, \tau \in S_n$ and study the effect of rotation on these permutations. We prove that rotation must necessarily be a special kind of transposition operation and hence we call it a \textit{tree transposition}.

A transposition is an operation on a permutation (in general on any sequence of numbers) that consists in swapping two consecutive sequences of the permutation. We define this formally now (cf. \cite{BFR11}). Let $\sigma$ be a permutation. Given three integers $i,j,k$ such that $1\leq i<j<k \leq n$, the result of \emph{transposition} operation $\delta_{i,j,k}$ applied to the permutation $\sigma$ is the following permutation (let $q(j)=k+i-j$). For all $t \in [n]$,
For any $1\leq t< i$, $\delta_{i,j,k}(t)=\sigma(t)$. For any $i\leq t <q(j)$, $\delta_{i,j,k}(t)=\sigma(t+j-i)$. For any $q(j)\leq t <k$, $\delta_{i,j,k}(t)=\sigma(t+j-k)$.
For any $k\leq t \leq n$, $\delta_{i,j,k}(t)=\sigma(t)$. The \textit{transposition distance} between two permutations is the minimum number of transposition operations that are needed to transform one permutation to another.
A transposition $\delta_{i,j,k}$ is a \textit{$1$-transposition} if $j=i+1$ or $k=j+1$.

Using the above definition, a transposition $\delta$ (we drop the subscripts $i,j,k$) is said to be a \textit{tree transposition} if there exists a pair of trees $T_1$ and $T_2$ with corresponding permutations $\sigma$ and $\tau$, and $d(T_1, T_2) = 1$ such that the transposition operation $\delta$ applied to $\sigma$ gives $\tau$. We prove the following characterization.\vspace*{-2mm}

\begin{theorem}
\label{thm:trasposition-treeperm}
A transposition $\delta_{i,j,k}$ is a tree transposition if and only if it is a $1$-transposition.
\end{theorem}

While the characterization does not lead to a new algorithm for computating rotation distance between two trees, we believe that it might be useful in proving $\NP$-compeleteness of the problem, given that the problem of transforming a given permutation to identity permutation (this is called the sorting problem for permutations) with minimum number of transposition operations is known to be $\NP$-complete~\cite{BFR11}.\vspace*{-2mm}

\paragraph{Tree Polynomials and Rotation:} Wiley and Grey~\cite{WG14} associated a bivariate polynomial to every  binary tree (see Section~\ref{sec:tree-polynomials} for a precise definition and a statement of the characterization). A natural approach to the rotation distance problem is to interpret them using the bivariate polynomials corresponding to the trees and interpreting the effect of rotation on the polynomials. However, we show that there are two distinct full binary trees which will result in the same polynomial. On the positive side, when restricted to rank $1$ trees (that is, skew trees), the association is an injection, and we derive a characterization for these polynomials (which we call the {\em skew polynomials}). We also explore the effect of rotation in this representation and observe that the algorithm for finding the skew rotation distance has a natural counterpart in terms of polynomials as well (which we describe for completeness).

\paragraph{Organization of the Paper:}
The rest of the paper is organized as follows. We discuss the connection between permutations, transpositions, and tree rotation in Section~\ref{sec:tree-permutations} and provide the characterizations mentioned above. We present our results about rank bounded rotation distance in Sections~\ref{sec:rank-1} and \ref{sec:rotation-rankbounded-reduction}. Results about height bounded rotation distance is discussed in the Section~\ref{sec:rotation-heightbounded-reduction}. Finally, we discuss the characterization of tree-polynomials and skew versions of the same in Section~\ref{sec:tree-polynomials}. We conclude with discussion and open problems in Section~\ref{sec:openproblems}.

\section{Preliminaries}
\label{sec:preliminaries}

The nodes of a binary tree are classified as two types: leaf and internal. A leaf node is one which does not have a child, and an internal node is one which does. An (internal) node is called a caret node if it has two children (i.e. both a left child and a right child). A binary tree is called a full binary tree if every node is either a caret node or a leaf (i.e. if every internal node has two children). A binary tree is called a skew tree if every caret node has at most one child which is also a caret node (a full binary tree is a skew tree if and only if every internal node has at least one child which is a leaf). Every full binary tree has an odd number of nodes ($2n+1$ for some integer $n$) and $n$ of them are internal nodes and $n+1$ of them are leaves. A binary tree on $n$ nodes can be made into a full binary tree on $2n+1$ nodes by attaching a leaf to every node which has only one child and attaching two leaves to every node with no children so that each node in the original tree becomes an internal node in the new tree; this perspective is useful for this paper as we always consider the number of internal nodes of a full binary tree.

The height of a binary tree is the number of internal nodes in a longest possible path from the root to some leaf. Alternatively, the height of a binary tree with one internal node is $1$, and inductively, the height of a binary tree on more than one node is the one more than the maximum of the height of the left subtree and the height of the right subtree.

A labeled binary tree is a binary tree where every node contains a label. A tree traversal is an algorithm for "visiting" each node of a binary tree, and the traversal of a labeled binary tree is an ordered list of the labels of each node where each node appears exactly once in the list. Three common traversals are (1) pre-order traversal, (2) in-order traversal, and (3) post-order traversal each of which is defined inductively. For a labeled binary tree with one node, all three traversals are the single element list with the label.

Inductively, the pre-order traversal of a labeled binary tree is the concatenation of three lists: (a) the singleton list with the label of the root, (b) the pre-order traversal of the left subtree, and (c) the pre-order traversal of the right subtree.

Inductively, the in-order traversal of a labeled binary tree is the concatenation of three lists: (a) the in-order traversal of the left subtree, (b) the singleton list with the label of the root, and (c) the in-order traversal of the right subtree.

Inductively, the post-order traversal of a labeled binary tree is the concatenation of three lists: (a) the post-order traversal of the left subtree, (b) the post-order traversal of the right subtree, and (c) the singleton list with the label of the root.

A rotation of a (labeled) binary tree at a node is an operation which produces a new (labeled) binary tree such that the in-order traversal of the original tree is the same as the in-order traversal of the new tree. A right (resp. left) rotation can be performed at a node "a" if the node has a left (resp. right) child "b"; let "C" be the (possibly empty) left (resp. right) subtree of "b", and "D" be the (possibly empty) right (resp. left) subtree of "b", and let "E" be the (possibly empty) right (resp. left) subtree of "a" (see Figure \ref{fig:TreeRotationPrelims}).

\begin{figure}[h]
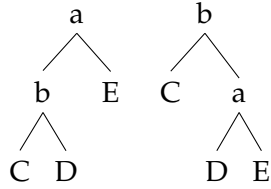

\vspace*{-3mm}
    \centering \hspace*{-5cm}
    \Tree[.a [.b [.C ] [.D ] ] [.E ] ] \hspace*{-3cm}  \Tree[.b [.C ] [.a [.D ] [.E ] ] ]
	  \caption{Right rotation at node "a"}
    \label{fig:TreeRotationPrelims}
\end{figure}

The new tree is identical except that the subtree at "a" is replaced with the following tree: the root is "b", the left (resp. right) child of "b" is "C", the right (resp. left) child of "b" is "a", the left (resp. right) child of "a" is "D", and the right (resp. left) subtree of "a" is "E".

The rotation distance between two full binary trees $T_1$ and $T_2$ is the minimum number of rotations required to transform $T_1$ into $T_2$ (or equivalently transform $T_2$ into $T_1$ since left and right rotations are inverse operations when applied at the appropriate nodes).

\medskip
\noindent{\bf Rank of Full Binary Trees:}
The rank of a binary tree was originally defined and studied in the context of decision trees~\cite{EH89} and hence is applicable to full binary trees. It is inductively defined as follows:
rank$(u) = 0$ if $u$ is a leaf. Otherwise $u$ is not a leaf, so it is internal and has two children $v$ and $w$, and the rank is defined as follows: rank$(u)$ = rank$(v)+1$ if rank$(v)$ = rank$(w)$ otherwise it is $\max\{$rank$(v)$,rank$(w)\}$.
The rank of the tree is the rank of the root node in the tree. It can also be seen that rank is exactly the maximum height of the complete binary tree that can be obtained from the given tree by a sequence of edge contractions (and hence a minor of the given tree, since deletion of vertices will disconnect the tree).

Using any of the above characterizations, we now prove that rank $1$ full binary trees exactly coincide with full binary skew trees.
\begin{proposition}
\label{prop:skewtree-rankbound}
For any full binary tree $T$, $rank(T) = 1 \iff T \textrm{ is a full skew tree}$.
\end{proposition}
\begin{proof}
We prove the forward direction. We can prove by induction on the height of the tree. For the base case, when the height is $1$. Let $rank(T)=1$. By definition, both children have rank $0$ and it is a skew tree. For the induction step, assume the statement for trees of height $h+1$. Since height is more than 1, it must be the case that one of the children has rank $0$ and the other has rank $1$. Thus the root has one child as a leaf (without loss of generality the left) and the other child is the root of a height $h$ tree of rank 1. Hence applying induction hypothesis implies that the right subtree of the root must be a skew tree which in turn implies that the tree $T$ by itself also must be a skew tree.

\medskip
The reverse direction follows directly by induction on the height of the tree. Base case can be checked easily. Let $T$ be a skew tree of height $h+1$. Without loss of generality assume that the left child of the root is a leaf. The right subtree of the root must be of height $h$, and hence by induction hypothesis must have rank at most 1. The rank of the root node of $T$ then must be $1$ by definition.
\end{proof}

\noindent {\bf Binary Trees and Triangulations:} There is a bijection between full binary trees with $n$ internal nodes and polygon triangulations of a polygon with $n+2$ sides which we outline below (see \cite{STT86} for more details). One of the edges in the polygon is marked as root. The other $n+1$ edges are labeled in counter clockwise direction. These corresponds to $n+1$ leaves of the full binary tree. The sides of the polygon are called edges and the chords that make triangles are called diagonals. The $n-1$ diagonals corresponds to $n-1$ internal nodes excluding the root. Sleator, Tarjan and Thurston  \cite{STT86} proved that rotation in the binary trees exactly correspond to an appropriate flipping of diagonal in the polygon representation, thus formulating the equivalent problem in terms of flipping distance between two triangulations of the polygons.

We now associate the common diagonals in the polygon representation to the corresponding nodes in the tree. For a diagonal $u$ in a polygon representation, let $L(u)$ and $R(u)$ represent the set of edges of the polygon on either side of the diagonal (with $L(u)$ containing the edge corresponding to the root of the tree). In the polygon representations of two trees $T$ and $T'$, a diagonal $u$ of the first polygon is said to be common with the a diagonal $v$ of the second polygon if $L(u) = L(v)$ and $R(u) = R(v)$. This can be equivalently interpreted in the trees also. Consider two full binary trees $T_1$ and $T_2$. A vertex $u$ in $T_1$ is said to be {\em common} with vertex $v$ in $T_2$, if in the in-order labeling of the leaves of the tree, the set of labels that appear in the subtrees rooted at $u$ and $v$ (in $T_1$ and $T_2$ respectively) are exactly the same.

The following lemma from Sleator, Tarjan and Thurston ~\cite{STT86} states  a property of shortest diagonal flip path between two polygon triangulations (equivalently, the shortest rotation sequence between the corresponding binary trees).
\begin{lemma}[\cite{STT86}]
\label{lem:slaetorlemma}
If $\tau_1$ and $\tau_2$ are two triangulations of the polygon (having $n+2$ vertices), having a diagonal in common, then a shortest (diagonal flipping) path from $\tau_1$ to $\tau_2$ never flips this diagonal. In fact, any path that flips this diagonal is at least two flips longer than a shortest path.
\end{lemma}

\section{Tree permutations and rotation}
\label{sec:tree-permutations}

As mentioned in the introduction, our main tool in the paper is a connection between trees and permutations and an interpretation of rotation in the corresponding permutation.  In this section, we develop these tools and then apply them in the upcoming sections.


\subsection{Characterization of tree permutations}

Recall the encoding of trees as pre-order traversals of the tree labeled with in-order labeling. Note that, not all permutations in $S_n$ will correspond to full binary trees. For example, consider the permutation $(2\ 1\ 5\ 6\ 4\ 3\ 7)$. When the in-order traversal is fixed to be $(1\ 2\ 3\ 4\ 5\ 6\ 7)$, this does not correspond to a full binary tree and hence is not a tree permutation, as follows from the upcoming Theorem \ref{thm:treepermutation-charactersization}.  We first characterize the tree permutations, proving Theorem~\ref{thm:treepermutation-charactersization} from the introduction.

\treepermutationcharacterization*
\begin{proof}
We prove the forward direction first. Suppose $\sigma \in S_n$ is a tree permutation. Since the tree permutation is the pre-ordering of the nodes whose in-order traversal is the identity permutation, for any $\sigma(i)$, all the nodes which are on its left subtree appears before any nodes in its right subtree, because of the pre-order. For any node $\sigma(i)$, all the values on the left subtree are smaller than $\sigma(i)$, and all the values on the right subtree are larger than $\sigma(i)$, because of the in-order. This ensures that no value lesser than $\sigma(i)$ appears after a value greater than $\sigma(i)$. In other words, there are no indices $i<j<k$ such that $\sigma(k) < \sigma(i) < \sigma(j)$.

\medskip
Now to prove the reverse direction, let $\sigma = a_1a_2\ldots a_n$, be the permutation under consideration.
We know that it satisfies the above avoidance property that there are no indices $i,j,k$ such that $i<j<k$ and $\sigma(k)<\sigma(i)<\sigma(j)$. Let us fix $i=1$ the first index. We consider the two following cases.
\begin{description}
\item{\bf Case 1: $\exists j > i$ such that $\sigma(j)>\sigma(i)$}: \\
Let $q = \min\{j \mid j > i \textrm{ and } \sigma(j)>\sigma(i) \}$.
We generate the tree as follows. The root of the tree is $\sigma(i)=a_1$. Nodes from $i+1$ to $q-1$ which are less than $\sigma(i)$ form the left subtree. Nodes from $q$ to $n$ forms the right subtree. The avoidance property ensures that there are no $k > q$ such that $\sigma(k) < \sigma(i)$. So all the elements on the right subtree are greater than $\sigma(i)$. If $q = i+1$, then the left subtree does not contain any internal node and hence will contain only the leaf node.
\item{\bf Case 2: $\forall j > i$, $\sigma(j) <\sigma(i)$} \\
In this case, the root of the tree is $\sigma(i)=a_1$. Right subtree is empty. Nodes from $i+1$ to $n$ which are less than $\sigma(i)$ form the left subtree.
\end{description}
The same procedure can be applied recursively on the permutation corresponding to the left and right subtree and produce the unique tree. This implies that any pattern avoiding permutation avoiding the above discussed pattern is a tree permutation. This completes the proof.
\end{proof}

\subsection{Characterizing rotations using transpositions}
\label{sec:rotdist-trasp}
In this section, we prove our results about rotation and transpositions that we described in section~\ref{sec:intro}. More specifically, we prove Theorem~\ref{thm:trasposition-treeperm} which concludes that
a transposition $\delta_{i,j,k}$ is a tree transposition if and only if it is a $1$-transposition. The plan of the proof is as follows:
We first prove that every tree transposition is also a $1$-transposition, Then we use a counting argument and show that the number of $1$-transpositions is equal to the number of tree transpositions.

\begin{lemma}
\label{lem-tree-to-one}
Every tree transposition is a $1$-transposition.
\end{lemma}

\begin{proof}
We know that a rotation operation affects the ordering of the subtree rooted at the point of rotation where as it does not affect the ordering at other parts of the tree.
Let $\sigma = AxByCDE$  where $A$ is the pre-order of the tree until $x$ which is the point of rotation ash shown in Figure \ref{fig:TreeTransposition}.

\begin{figure}[h]
    \centering
 \hspace*{-6cm}  \Tree[.A  [.x [.B ] [.y [.C ] [.D ] ] ] [.E ] ]~~
	  \caption{Proof of Lemma~\ref{lem-tree-to-one}}
    \label{fig:TreeTransposition}\vspace*{-2mm}
\end{figure}

$E$ is the right sub-tree of $A$. $B$ corresponds to the pre-order of the left subtree of $x$. $y$ is on the right of $x$, and the left and right subtrees of $y$ are $C$ and $D$ respectively. We get $\sigma'= AyxBCDE$ by applying a left rotation at $x$. We observe that only one element ($y$) gets shifted, while the relative ordering of others remain the same. Hence this operation correspond to a $1$-transposition.
\end{proof}

Now we prove the equivalence by way of counting the number of tree transpositions and $1$-transpositions separately (see Lemma~\ref{ABOSPcount} and Lemma~\ref{lem:1-transposition-count}) and hence conclude the proof of Theorem~\ref{thm:trasposition-treeperm}.

\begin{lemma}
\label{ABOSPcount}
The number of tree transpositions on $n$ elements is $(n-1)^2$.
\end{lemma}
\begin{proof}
We saw that a tree transposition is obtained from a tree permutation by shifting exactly one element $i$ to any other position such that $\sigma(i) \neq i$. There exists tree permutations of length $n$, for which tree transpositions are possible for all $1\leq i \leq n$. The number of ways in which each element can be shifted to form such a pattern is $n-1$. So, for $n$ elements, the total count is $n(n-1)$. But this is over counting. We counted separately on two adjacent elements $i,j$, one for shifting $i$ to the right of $j$ and then for shifting $j$ to the left of $i$. By subtracting the extra counted number of adjacent pair elements we get the total count as $n(n-1) - (n-1) = (n-1)^2$.
\end{proof}

\begin{lemma}
\label{lem:1-transposition-count}
The number of 1-transpositions on $n$ elements is $(n-1)^2$.
\end{lemma}
\begin{proof}
A 1-transposition can be viewed as moving an element at index $i$(resp. $j$) to an element after (resp. before) index $j$(resp. $i$), where $i<j$. Let us name the two cases as {\em left-right} and {\em right-left} respectively. Now, let us count the number of {\em left-right} and {\em right-left} 1-transpositions.

Let the indices be numbered from $1$ to $n$.
For {\em left-right}, after fixing $i$, the $j$ value can range from $i+1$ to $n$. So, for $i=1$, $j$ can take any of the $n-1$ values, for $i=2$, $j$ can take any of the $n-2$ values$\ldots,$ for $i=n-1$, $j$ can take  only one value (which is $n$). Thus there are $\frac{n(n-1)}{2}$ {\em left-right} cases. Similarly, same number of {\em right-left} cases by a similar argument. The total number of both adds up to $n(n-1)$

\medskip
We see that when $i,j$ differ by $1$, we can consider it as either moving the left element to the right of the next element or moving the right element to the left of the previous element. The number of such possibilities are $n-1$. Now by subtracting this over count, we get $n(n-1)-(n-1)$ which is $(n-1)^2$.
\end{proof}

\subsection{Characterization of skew permutations}

Recall that skew trees are binary trees where each caret node has a leaf as one of the children (in particular, a full binary tree is a skew tree if and only if
every internal node has at least one child which is a leaf ). We now show characterization for the tree permutations corresponding to skew trees.
The argument is a simple induction on the height of the tree by using the structural property of skew trees.

\introlemmaminmaxskew*

\begin{proof}
We prove the theorem by induction on the height of the tree which is the maximum number of internal nodes in any root to leaf path. As the base case, when the height is $1$, $T$ consists of a single  node, $\sigma (1) = 1 = \min/\max \{1\}$. For the induction, assume that the statement is true for all trees of height $h$, we prove it for height $h+1$. Let $T$ be a tree of height $h+1$.
\begin{description}
\item{\bf Case 1}: When root node in $T$ has a non empty skew tree $T_L$ as the left child. We apply induction hypothesis to $T_L$. Let $\ell = |T_L|$ be the number of internal nodes in $T_L$. By the definition of in-order labeling, root node gets the label $\ell+1$. $\sigma (1) = \ell+1 = \max \{\sigma(1), \sigma(2), \ldots , \sigma(\ell), \sigma(\ell+1)\}$. Since $n = \ell+1$, the tree $T$ satisfies the property.

\item{\bf Case 2}: When root node in $T$ has a non empty skew tree $T_R$ as the right child. Let $\ell = |T_R|$ be the number of internal nodes in $T_R$. By the definition of in-order labeling, root node gets the label $1$ and labels from $2$ to $\ell+1$ will be assigned among the nodes in $T_R$. Hence $\sigma(1) = 1 =  \min \{\sigma(1), \sigma(2), \ldots \sigma(\ell), \sigma(\ell+1)\}$. Since $n = \ell+1$, the tree $T$ satisfies the property.
\end{description}
Since the cases were exhaustive, we have completed the proof.
\end{proof}


Now we define a relation between two permutations: 
Let $\sigma, \tau$ be two permutations on $n$ elements. $\tau$ is said to be the result of a \textit{skew transposition} of $\sigma$ (denoted $\sigma \sim \tau$) if there is a $i \in [n]$ such that (1)
$\sigma(i) = \min/\max~\{\sigma(i),\sigma(i+1)\ldots \sigma(n)\}$.
(2)~$\sigma(i+1) = \min/\max~\{\sigma(i),\sigma(i+1)\ldots \sigma(n)\}$
(3) for all $j \in [n]$:
$\tau(j) = \sigma(i)$ when $j= i+1$, $\tau(j) = \sigma(i+1)$ when $j = i$ and  $\tau(j) = \sigma(j)$ for all other $j \in [n] \setminus \{i,i+1\}$. We relate skew transpositions with rotation distance for skew trees.

\begin{lemma}
\label{minmaxswap}
Let $T_1, T_2$ be full skew trees with associated tree permutations  $\sigma$ and $\tau$ respectively:
$d(T_1, T_2) = 1 \textrm{ if and only if } \sigma \sim \tau$.
\end{lemma}

\begin{proof}
We prove the forward direction first. Let $d(T_1, T_2) = 1$. We prove that the corresponding tree permutations $\sigma$ and $\tau$ satisfy $\sigma \sim \tau$.
We know that $T_1$ and $T_2$ are skew trees and hence by Theorem \ref{minmaxSkew}, $\forall i, \sigma(i)=\min/\max\{\sigma(i),\sigma(i+1),\ldots, \sigma(n)\}$. From the property of skew permutations it is known that if $\sigma(i)=\min\{\sigma(i),\sigma(i+1),\ldots, \sigma(n)\}$ then $\sigma(i+1)=\max\{\sigma(i+1),\sigma(i+2),\ldots, \sigma(n)\}$ or vice versa. We verify the two cases below.
Let $\tau$ is the  traversal of $T_1$ obtained by a single flip of a consecutive min-max or max-min pair. We call an index $i$ to be a max-index if the $\sigma(i)$ in the characterization is a max operation and otherwise we call it a min-index.

\begin{description}
\item{\bf Case 1: $i$ is a min-index and $i+1$ is a max-index:}
By definition, we know that, $\sigma(i) = \min\{\sigma(i)$, $\sigma(i+1)\ldots \sigma(n)\}$ and $\sigma(i+1) = \max\{\sigma(i+1), \sigma(i+2)\ldots \sigma(n)\}$
It can be seen that $\tau(i)$ becomes $ \sigma(i+1)$ and  $\tau((i+1))$ becomes $\sigma(i)$.
After the rotation, in $T_2$, $\sigma(i)$ will become the root of left subtree of $\sigma(i+1)$ in $T_2$ and $\sigma(i+2)$ will become the root of right subtree of $\sigma(i)$.

With respect to $T_2$,
$\tau(i)  = \sigma(i+1) = \max\left\{\sigma(i+1),\sigma(i+2)\ldots \sigma(n)\right\}$. This, in turn is:
$\max\left\{\sigma(i),\sigma(i+1)\ldots \sigma(n)\right\} = \max\left\{\tau(i),\tau(i+1)\ldots \tau(n)\right\}$.

Similarly, $\tau(i+1) = \sigma(i) = \min\{ \sigma(i), \sigma(i+1) \ldots \sigma(n) \} \setminus \{\sigma(i+1)\} \}$
 Since $\sigma(i) < max\{\sigma(i+2), \sigma(i+3)\ldots \sigma(n)\}$. This, in turn is $\min\{  \tau(i+1), \tau(i+2) \ldots \tau(n)\}$.

\item {\bf Case 2: $i$ is a max-index and $i+1$ is a min-index.}
We know that $\sigma(i) = \max\{\sigma(i), \sigma(i+1)\ldots \sigma(n)\}$ and
$\sigma(i+1) = \min\{\sigma(i+1), \sigma(i+2)\ldots \sigma(n)\}$.

By similar argument as in {\sf Case 1} $\tau(i)=\sigma(i+1)=\min\{\sigma(i+1),\sigma(i+2)\ldots \sigma(n)\}$ and $\tau(i+1)=\sigma(i)=\max\{\sigma(i),\sigma(i+1)\ldots \sigma(n)$.

\end{description}
We see that in both cases
that the subtree rooted at $\sigma(i+2)$ does not change. It can also be noted that the first $(i-1)$ elements also does not change.
Thus in both cases we get a skew tree and hence $\sigma \sim \tau$.

Now, we prove the reverse direction. Let $\sigma \sim \tau$
and let $T_1,T_2$ be the skew trees corresponding to $\sigma\text{ and } \tau$ respectively.
The key point here is that the first two properties of the skew transposition ensures that if $\sigma(i) = \min\{\sigma(i),\sigma(i+1),\ldots, \sigma(n)\}$ then $\sigma(i+1) = \max\{\sigma(i),\sigma(i+1),\ldots, \sigma(n)$ or vice versa.
From the property (3) of skew transposition,
we have that for all $j \in [n]$:
This ensures that only the adjacent elements at $i,i+1$ gets swapped among $\sigma$ and $\tau$. As we saw in the proof of forward direction $\tau(i) = \sigma(i+1)$ and $\tau(i+1) = \sigma(i)$. It also ensures that if $\tau(i)= \min\{\tau(i),\tau(i+1),\ldots, \tau(n)\}$ then $\tau(i+1)= \max\{\tau(i+1),\tau(i+2),\ldots, \tau(n)$ or vice versa.
If $\sigma(i)$ and $\sigma(i+1)$ are either  $\min/\max \{\sigma(i),\sigma(i+1),\ldots, \sigma(n)$ but not both, then $\tau(i)$ and $\tau(i+1)$ also holds the same property.

We prove that indeed this corresponds to a single skew rotation. Let $\sigma(i)= \min\{\sigma(i),\sigma(i+1),\ldots, \sigma(n)\}$ and $\sigma(i+1)= \max\{\sigma(i),\sigma(i+1),\ldots, \sigma(n)\}$. We now construct a  tree rooted at $\sigma(i)$ satisfying $\sigma$. It has $\sigma(i+1)$ as the root of its right subtree and $\sigma(i+1)$ has $\sigma(i+2)$ as its root of its left subtree. From the definition we know that $\sigma$ corresponds to a skew tree. Let us try to do a left rotation on the tree at $\sigma(i)$. The new tree has $\sigma(i+1)$ as root of a subtree whose left subtree is rooted at $\sigma(i)$ whose right subtree is rooted at $\sigma(i+2)$. The permutation corresponding to this new tree is $\sigma(1)\ldots \sigma(i+1)\sigma(i) \sigma(i+2)\ldots \sigma(n)$. This corresponds to a new permutation $\tau$ where $\tau(i)=\sigma(i+1)$ and $\tau(i+1)=\sigma(i)$.

\medskip
Thus $\sigma,\tau$ corresponds to  skew permutations of a pair of trees that are one rotation distance apart. Hence $d(T_1,T_2) = 1$.
 Similar argument can be used to prove the case where $\sigma(i)= \max\{\sigma(i),\sigma(i+1),\ldots, \sigma(n)\}$ and $\sigma(i+1)= \min\{\sigma(i+1),\sigma(i+2),\ldots, \sigma(n)\}$.
\end{proof}

\noindent {\bf Encoding using a binary string: } We have seen from Theorem \ref{minmaxSkew} that the pre-order traversal of a skew tree with $n$ nodes is precisely characterized by the condition:  $\forall i , \sigma(i) = \min/\max\{ \sigma(i),\sigma(i+1),\ldots, \sigma(n)\}$. Thus, a natural encoding of such a permutation is by considering a string with $0$s at all the min-indices of $\sigma$ and $1$ at all the max indices. Because at the final index n, we have $\sigma(n)=min{\sigma(n)}=max{\sigma(n)}, $ for we can always use either $0$ or $1$ as the last bit of the encoding. As a convention, we will require the encoding to have $a_n$ be the opposite bit of $a_{n-1}$.

\section{Rank $1$ Case: Computing skew rotation distance}
\label{sec:rank-1}

In this section, we present an efficient algorithm for computing the skew rotation distance between two given skew trees to prove the following theorem from the introduction.

\skewtreesalgo*

\begin{proof}
 Let $T_1,T_2$ be the two given skew trees having $n$ internal nodes each. We use the binary encoded representation for them.
 Let $a = a_1a_2\ldots a_n$ and $b = b_1b_2\ldots b_n$ represent $T_1,T_2$ respectively. Let  the two strings be of length $n$ which equals the number of internal nodes in each tree. We propose the following algorithm (see Algorithm~\ref{algo:binskewrot}).
\begin{algorithm}
\KwIn{Two binary strings $a$ and $b$ of length $n$ representing full skew trees $T_1$ and $ T_2$} respectively with $n$ internal nodes each.\\
\KwOut{count - the skew rotation distance between $T_1$ and $T_2$.}
Set count $=0$;

\While{$a \neq b$}{
  Find the smallest $i$ where $a_i \neq b_i$;

  Find the smallest $j>i$ where $a_j = b_i$;

  \While{$j \neq i$}{
        Swap $a_j$ and $a_{j-1}$ \\
  	$j = j-1$ \\
  	\textrm{count $=$ count$+1$} \\
  }
  $a_n = a_{n-1} \oplus 1$
}
\Return{\tt count}\;
\caption{{\sc BinarySkewRotation}}
\label{algo:binskewrot}
\end{algorithm}

\noindent Note that operation $a_n = a_{n-1} \oplus 1$ also ensures that $a$ is as per our convention, i.e.  $a_{n-1}\neq a_n$, in particular, when $j=n-1$ in line 10. We also observe that step 4 and 5 are well-defined: Since $T_1 \neq T_2$ (equivalently $a \neq b$), there exists an index $i$ for which $a_i \neq b_i$. Since the last two bits of the encoding are always ensured to be different, there exists $j \neq i$ such that $a_j = b_i$. \\[-3mm]

\noindent {\bf Termination and Runtime Bound:}  We argue termination first.
We argue that after each iteration, the number of indices $i$, such that $a_i = b_i$ strictly increases. Note that the length of the strings remain the same throughout the operation.
$a_{n-2}= 0$ and $a_{n-1} =1 \implies a_n=0$ (similarly, $a_{n-2}=1$ and $a_{n-1}=0 \implies a_n=1$). A swap between $a_{n-2}$ and $a_{n-1}$  represent the new tree with $a_{n}\neq a_{n-1}$. Since $a_n \neq a_{n-1}$ is the condition to be satisfied, $a_n$ is changed to $1$ during the swap operation when $a_{n-1}=0$ and vice versa. Thus, a swap operation involving $n-2$ and  $n-1$ increases or decreases the number of $0$s and $1$s to match with the number of $0$s and $1$s in both binary strings. Since $a_n \neq a_{n-1}$ it is always possible to find an index $j$ such that $a_j = b_i$. The number of swap operations to match a single bit $a_i = b_i$ is at most $n-1$, i.e. $O(n)$. Since the length of the string is $n$, the running time of the algorithm is $O(n^2)$.\\[-3mm]

\noindent{\bf Correctness:}
The main observation is that a rotation at a node other than an angle node generates a non-skew tree. Hence, the optimal rotation sequence that achieves the above value of $k$ can contain only rotations that correspond to rotation at an angle node. In terms of the encoding, this exactly corresponds to swapping a pair of adjacent symbols that are different (either a $01$ or a $10$). Conversely, any swap between adjacent $0$-$1$ pair (or $1$-$0$ pair) indeed corresponds to a skew rotation in the respective tree. Thus, it suffices to argue that the algorithm finds the optimal number of swaps (between adjacent $0$-$1$ pairs or adjacent $1$-$0$ pairs) to transform $a$ to $b$.

\medskip
In the algorithm, line 5-9, targets to make $a_i = b_i$, by finding the smallest index $j$ for which $a_j = b_i$ and shifts $a_j$ to the left left by $j-i$ number of $0$-$1$ swaps. Indeed, any optimal sequences that attempts to make $a = b$ must also necessarily make $j-i$ number of $0$-$1$ (or $1$-$0$) swaps since $j$ is the smallest index such that $a_j = b_i$. Hence the sequence of swaps that the algorithm produces must necessarily be an optimal sequence in terms of the number of swaps.
\end{proof}

\noindent {\bf Distance to nearest skew tree:} We derive the following additional information about how the skew trees are placed in the rotation graph.
\begin{proposition}
\label{prop:dist-skew-tree}
The rotation distance of a full binary tree  $T$  with $n$ internal nodes and height $h$, to its nearest skew tree is exactly $n-h$.
\end{proposition}

\begin{proof}
Let $S_1, S_2, \ldots, S_n$ represents the set of internal nodes at levels $1, 2, \ldots, n$ respectively of a full binary tree T. The number of sets  $S_i\neq \phi, 1\leq i \leq n$, is equal to the height $h$ of the tree. In other words $h = $max$(i)|~S_i \neq \phi$. Any set $S_i \neq \phi$ implies that $\forall j, 1\leq j<i, S_j \neq \phi$.
A skew tree is one in which all the nodes are at different levels. Hence for a skew tree $S_n \neq \phi$.  Since there are $n$ internal nodes it is straight forward to see that for a skew tree, $\forall i, 1\leq i \leq n, |S_i| = 1$.
We know that by a single rotation, the height of a tree changes at most by one. Similarly the number of non-empty sets $S_i$ can be increased at most by one by a single rotation. So, if $S_1, S_2,\ldots, S_h$ are the non-empty sets of nodes at each level, it requires at least $n-h$ rotations so that $S_1, S_2,\ldots, S_n$ are all non-empty.

\medskip
We now prove that $n-h$ is indeed the upper bound for the number of rotations to be made to the nearest skew tree. For this, it is sufficient to argue that for every tree, there is a rotation which strictly increases the height of the tree (thus, we can increase the height to get to a height $n$ tree, which is necessarily a skew tree). The idea is to do a rotation at the first non-skew descendent of root. More precisely, let $T$ be a tree on $n$ internal nodes whose height of left subtree and right subtree is represented by $h_L, h_R$ respectively. We do a left (right) rotation at root of $T$ if $h_L\geq h_R~(h_L<h_R)$. This ensures that the height of $T$ gets incremented by $1$ because the value of $h_L~(h_R)$ gets incremented by $1$ whereas that of $h_R~(h_L) $ gets decremented by $1$. The same is applied on $T$ till $h_R~(h_L)$ becomes zero (i.e. the right (left) subtree is a leaf). Now, recursively the same procedure is done starting from left (right) subtree of the $T$. Since at each rotation the height of $T$ gets incremented by 1, a tree of height $n$ is obtained which is necessarily skew. The number of rotations made is at most $n-h$ where $h$ is the height of the initial tree $T$.
\end{proof}

\noindent {\bf Distance to the right comb:} Given the above proposition, it is natural to ask the exact distance between a tree and the special skew trees\textemdash the right and the left comb. The following is an easy proposition using a similar idea as above, which we include with proof, for completeness.

\begin{proposition}
\label{prop:dist-rightcomb-tree}
The rotation distance of a full binary tree  $T$  with $n$ internal nodes and $r$ internal nodes in the rightmost path, to the right comb on $n$ internal nodes, is exactly $n-r$.
\end{proposition}
\begin{proof}
We first prove the lower bound. Let $T$ be the given tree. Let $r$ be the number of internal nodes in the rightmost path in the tree $T$. The right comb has exactly $n$ internal nodes on the rightmost path from root to leaf of the tree. By a single rotation, the number of nodes on the right most path can be increased at most by one (and this is possible only if it is a right rotation with respect to a node in the rightmost path). Hence, any sequence of rotations which transforms $T$ to the right comb must do at least $n-r$ rotations.

\medskip
Now we argue the upper bound. Indeed, the statement follows when $r=n$, since there is only one tree (that is the right comb) which has all $n$ internal nodes in the rightmost path from root to leaf. Consider the case $r<n$; there exists at least one internal node on the right path whose left child is an internal node. A right rotation at this node increases the number of nodes on the right path by $1$ (hence $r$ increases by $1$). By repeating this process, we will reach the tree with $r=n$, which is the right comb.
\end{proof}

\section{Rank bounded paths in the rotation graph}
\label{sec:rotation-rankbounded-reduction}

As described in the introduction, the {\sc RotDist$_r$ } problem asks the following: given two trees $T_1, T_2$ and an integer $k$, with $r = \max\{\rank(T_1), \rank(T_2)\}$, is it possible to covert $T_1$ to $T_2$ with at most $k$ rotations such that intermediate trees obtained during the rotation are all of rank at most $r$? We show now that the general rotation distance problem can be reduced to this version of the problem, thus proving Theorem~\ref{thm:rankboundedrotationdistancereductiontheorem} from the introduction (restated next).

\rankboundedrotationdistancereductiontheorem*
\eject

\begin{proof}
Given two trees $T_1$ and $T_2$ on $n$ internal nodes each, we convert them to $T_1', T_2'$ respectively as follows.

\smallskip
\begin{center}
    \includegraphics[scale=0.6]{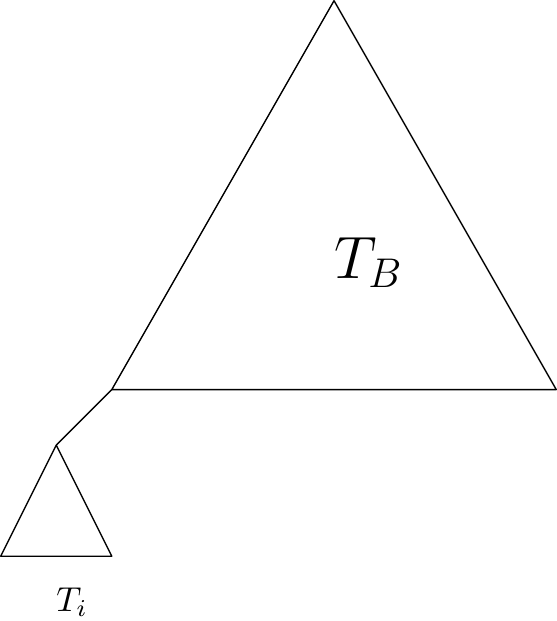}\vspace*{-2mm}
\end{center}
For each $T_i$ ($i \in \{1,2\}$), consider a complete binary tree $T_B$ (which we will refer to as the {\em gadget tree}) of rank $r= \lceil \log_2 n \rceil +1$. $T_i'$ is the tree obtained by attaching $T_i$ to the leftmost leaf $\ell_i$ of the complete binary tree.

\medskip
Note that the rank of these new trees are equal to $r$ irrespective of the ranks of the original trees. Indeed, for $i \in \{1,2\}$ the rank$(T_i')\leq r$ because for it to be $r+1$, we need at least $2^{r+1}$ nodes, but we have only $n+2^r$ nodes in the constructed tree. Also, rank$(T_i')\geq r$ because there is a complete binary tree of height $r$ which is a minor of $T_i'$.
We now prove that, for any $k \in \mathbb{N}$ $d(T_1,T_2) \le k$ if and only if $d_R(T_1',T_2') \le k$.

The forward direction follows from the fact that any rotation of the trees $T_1$ and $T_2$ can be implemented as a rotation inside the subtree rooted at the leftmost leaves $\ell_1$ and $\ell_2$ (in trees $T_1'$ and $T_2'$ respectively). By the argument in the previous paragraph, the rank of the trees that appear in the intermediate stages remains exactly $r$.

To prove the reverse direction, assume $d_R(T_1',T_2') \leq k$. That is, there is a sequence of rotations that transforms $T_1'$ to $T_2'$ such that the intermediate trees in the process have rank at most $r$. By definition, $d(T_1',T_2') \leq k$. Consider the shortest path between $T_1'$ and $T_2'$ in the rotation graph. We use a structural property of the shortest paths in the rotation graph described by Sleator, Tarjan and Thurston \cite{STT86} (see Lemma~\ref{lem:slaetorlemma}). To use this property, we observe that all the internal nodes in the gadget tree $T_B$ in both the trees are common vertices between $T_1'$ and $T_2'$ since the leftmost tree will contain all the $n$ internal nodes of the original tree.

\medskip
By applying Lemma~\ref{lem:slaetorlemma}, we conclude that there is a shortest path between $T_1'$ and $T_2'$ in the rotation graph which does not apply rotation with respect to the internal nodes in the gadget tree $T_B$. Hence, the sequence of rotations demonstrated by such a shortest path can directly be applied on the trees $T_1$ and $T_2$ as well, thus showing that $d(T_1, T_2) \le k$.
\end{proof}

Now we prove the upper bound on the rank-bounded rotation distance mentioned in the introduction.

\rankdistbound*
%

\begin{proof}
We start with a special case of the above bound where $r_1 = r_2 = 1$. That is, $T_1$ and $T_2$ are skew trees. We show a distance upper bound in this case.
\begin{lemma}
\label{prop:skewrotationdistance}
For any two skew trees $T_1$ and $T_2$, we have that $d_R(T_1,T_2) \le n^2$.
\end{lemma}

\begin{proof}
Let $a,b$ be the two $n$ length binary strings representing $T_1,T_2$. We know that in order to make any $a_i=b_i$, the maximum number of swaps (corresponding to skew rotation) that needs to be done is at most $n$. Since the number of bits in the string is $n$, the maximum number of swaps, there by the maximum number of skew rotations to be performed is at most $n^2$.
\end{proof}

Towards proving the upper bound in the general case, we establish an upper bound on the number of rotations required to reduce the rank by $1$.
\begin{lemma}
\label{lem:n2-rotations-to-reduce-rank}
For any full binary tree $T$ on $n$ internal nodes with rank $r$, $r\geq 2$, we can reduce the rank of $T$ to $r-1$ by using at most $2n^3+n^2$ rotations, such that all the intermediate trees obtained in the sequence also have rank at most $r$.
\end{lemma}

\begin{proof}
In order to prove the upper bound, we show a rotation sequence which satisfies the required constraints. We find the rank 2 nodes in each path of the tree which are farthest from the root and reduce the rank of the subtree rooted at them to 1 by appropriate operations. We also ensure that during the operation none of the intermediate trees generated has a rank more than the initial rank of the tree.

At each rank 2 node $v$ in each path of the tree which are farthest from the root, we do this two step process: (1) apply Lemma~\ref{prop:skewrotationdistance} at $v$ to convert the left-subtree to the left-comb and the right-subtree to the right-comb. (2) perform repeated right rotations at node $v$ to convert subtree rooted at $v$ to a right comb. Note that step (1) does not change the rank of $v$ because the left and right subtree of $v$ always are maintained to be skew trees. And after step (1) is done, the subtree at $v$ has the property that the left subtree  is a left-comb and the right subtree is a right-comb (which we will call the left-right-comb property). Step (2) ensures that the tree still has the left-right-comb property until the tree rooted at $v$ itself becomes a right comb.

Let $T$ be a tree of rank $r\geq 2$. We now prove that the above procedure ensures reduction in rank of $T$ by 1.  Consider the sets $S_1,\ldots,S_r$ where $S_i$ is the set of internal nodes of rank $i$. We know that when $T$ is of rank $r$, by definition $S_i\neq \phi, 1\leq i \leq r$. By applying steps (1) and (2) for all nodes in $S_2$ that are farthest from the root in each path, let the tree so obtained be $T'$. Now let $S_1',\ldots,S_r'$ where $S_i'$ is the set of internal nodes of rank $i$ in $T'$. Since all nodes in $S_2$ are changed to rank 1 by the above two steps, it can be easily verified that $S_1' = S_1 \cup S_2$. Now consider the nodes in $S_3$. In $T'$, all nodes of rank 2 in $T$ were changed to rank 1. Thus, all nodes in $S_3$ have rank 2 in $T'$ and forms $S_2'$. By the same argument, $\forall i \le r$, $S'_{i-1} = S_i$. Indeed, this ensures $S_r' = \phi$ which makes the rank of $T'$ to be of rank $r-1$, i.e. rank($T'$) = rank($T$)$-1$.

We now bound the number of rotations performed in the above process. For the given tree $T$ of rank $r$, let $\ell = |S_2| \le n$. For each $v \in S_2$, the step (1) of the above process (Lemma~\ref{prop:skewrotationdistance}) costs at most $n^2$ rotations for the two children, and then step (2) costs at most $n$ rotations. Thus, the number of rotations performed is at most $2n^2+n$ for processing the vertex $v$. In total, the number of rotations is at most $\ell(2n^2+n) \le 2n^3+n^2$. This completes the proof of the Lemma~\ref{lem:n2-rotations-to-reduce-rank}.
\end{proof}

We will now prove Theorem~\ref{thm:rankdistbound}.
By multiple applications of Lemma \ref{lem:n2-rotations-to-reduce-rank} we can conclude that any full binary tree $T_1$ on n
internal nodes with rank $r_1$ can be transformed into a tree $T_1'$ of rank 1 using at most
$(r_1 - 1)(2n^3 + n^2)$ rotations with no intermediate tree having rank exceeding $r_1$ (and similarly for
$T_2, r_2, T_2'$). By Lemma \ref{prop:skewrotationdistance}, we can transform between
any two trees $T_1'$ and $T_2'$ of rank $1$ each on $n$ internal nodes each with at most $n^2$ rotations with all
intermediate trees having rank $1$. Thus,
we can transform from $T_1$ to $T_1'$, then $T_1'$  to $T_2'$ and finally from $T_2'$ to $T_2$ in at most
$n^2 + (r_1 - 1)(2n^3 + n^2) + (r_2 - 1)(2n^3 + n^2)$ rotations where every intermediate tree having
rank at most max(rank$(T_1)$, rank$(T_2)$). Simplifying gives the final expression that $d_R(T_1,T_2) \le n^2 (1+(2n+1)(r_1+r_2-2))$ where $r = \max\{r_1,r_2\}$. This completes the proof of the theorem.
\end{proof}

We remark that Theorem~\ref{thm:rankdistbound} is asymptotically tight for $r=1$. We show examples of two skew trees where the shortest path between them in the rotation graph is of length $\Omega(n^2)$. We use the characterization from the previous sections to prove this.

\begin{proposition}
\label{prop:skewrotdistance}
There exists two full skew trees $T_1$ and $T_2$ with $n$ internal nodes each, such that their skew rotation distance is at least $\frac{n(n-1)}{2}$
\end{proposition}

\begin{proof}
Let $T_1, T_2$ be the right-comb and left comb respectively of $n$ internal nodes each. As per our convention we know that the binary strings corresponding are $a = 0^{n-1}1$ and $b = 1^{n-1}0$ respectively. In order to ensure that the intermediate trees along the shortest rotation distance path are of rank at most $1$, rotation is possible at node $i$, if $\sigma(i)$ = min$\{\sigma(i),\ldots, \sigma(n)$ and $\sigma(i+1)$ = max$\{\sigma(i+1),\ldots, \sigma(n)$ or vice versa. Algorithm \ref{algo:binskewrot} can be used to find the skew rotation distance between $T_1$ and $T_2$. In order to make $a_1 = b_1$, it can be verified that the number of flips  that corresponds to a rotation is $n-1$. Similarly, in order to make $a_2 = b_2$, $n-2$ rotations are necessary. Continuing this way, we see that the total number of rotations required to make $a=b$ is $n-1+\ldots +1 = n(n-1)/2$.
\end{proof}

\section{Height restricted paths in the rotation graph}
\label{sec:rotation-heightbounded-reduction}
In this brief section, we show that a similar argument to Theorem~\ref{thm:rankboundedrotationdistancereductiontheorem} described in the previous section, can also be used to establish that another structurally restricted instance of the problem is also as hard as the general case.

\medskip
We define the restriction first.
Consider two full binary trees $T_1$ and $T_2$ of $n$ internal nodes each (of height at most $h$). We define a path from $T_1$ to $T_2$ in the rotation graph (${\cal R}_n)$ to be a \textit{height-bounded path} if each intermediate tree that appears in the path has height at most $h$. We define the height bounded rotation distance, $d_H(T_1,T_2)$, to be the length of the shortest height bounded path between the two trees $T_1$ and $T_2$ of height at most $h$. The $\HROTDIST$ problem is defined as:
\ourproblem{$\HROTDIST$: Given $h$, $k$ and two full binary trees $T_1$ and $T_2$ of height at most $h$, if the height bounded rotation distance between them is at most $k$ or not. That is, is $d_H(T_1,T_2) \le k$? }

Notice that by one rotation, the height of the tree can change at most by 1.

\heightboundedrotationdistancereductiontheorem*

\begin{proof}
Given two $n$ node trees $T_1,T_2$, we convert them to $T_1',T_2'$ respectively as follows. For each $T_i$ ($i \in \{1,2\}$), consider a right comb of height $n+1$. We obtain $T_i'$ by attaching $T_i$ as the left subtree of the root of right comb. The height of the tree remains bounded above by $n+1$, as long as rotation with respect to the root of $T_i'$ is not performed because the total number of internal nodes in $T_i$ is $n$.

\medskip
We now prove that, for any $k \in \mathbb{N}$, $d(T_1,T_2) \le k$ if and only if $d_H(T_1',T_2') \le k$. We use a structural property of the shortest paths in the rotation graph described by Sleator, Tarjan and Thurston \cite{STT86} (see Lemma~\ref{lem:slaetorlemma}) to prove the theorem. We observe that  all internal nodes in the right comb of both $T_1'$ and $T_2'$ are common vertices. The forward direction follows from the fact that any rotation of the trees $T_1$ and $T_2$ can be implemented as a rotation on the left subtrees of $T_1'$ and $T_2'$ respectively. By the argument in the previous paragraph, the height of the trees that appear in the intermediate stages remains exactly $n+1$, as the right comb part of the trees in the intermediate stages remains unaltered and thus $d_H(T_1',T_2') \leq k$.

\medskip
To prove the reverse direction, assume $d_H(T_1',T_2') \leq k$. That is, there is a sequence of rotations that transforms $T_1'$ to $T_2'$ such that the intermediate trees in the process have height at most $n+1$. Note that the number of intermediate trees in the process is less than $k$. By using the lemma described by Sleator, Tarjan and Thurston, there exists a sequence of at most $k$ rotations in the rotation graph that transforms $T_1'$ to $T_2'$ in which all the intermediate trees have the common right comb. Hence the sequence of rotations demonstrated by such a shortest path can directly be applied to trees $T_1$ and $T_2$ as well, thus showing that $d(T_1,T_2) \leq k$.
\end{proof}

\section{Tree polynomials and rotation distance}
\label{sec:tree-polynomials}

As mentioned in the introduction, with every  full binary tree $T$, we associate a bivariate polynomial $p_T(x,y) \in \mathbb{R}[x,y]$. Each edge of the tree is labelled with the variable $x$ and $y$ as follows: for an edge $(u,v)$ (with $u$ as the parent and $v$ as the child), we label it with $x$ (resp. $y$) if $v$ is the left child (resp. right child) of $u$. A path in the tree naturally gives rise to a monomial which is the product of the labels (each are variables $x$ or $y$) of edges in the path. The polynomial associated with the  full binary tree $T$ is the sum of the monomials corresponding to each root to leaf path.

A bivariate polynomial $p(x,y)$ is called a {\em tree polynomial} if there is a full binary tree $T$ such that $p(x,y)$ is exactly $p_T(x,y)$. This was studied by Wiley and Grey~\cite{WG14} where they described a necessary and sufficient condition for a polynomial to be a tree polynomial.

\medskip
We describe their result in this context for completeness. Let $p(x,y)$ be a bivariate polynomial with total degree at most $d$, of the following form: $p(x,y) = \sum_{\substack{a,b \in \mathbb{N}\\ 0 \le a,b \le d}} c_{ab}~x^ay^b$.
For example the polynomial $5x^2+7xy+2y^2+y+6$ would result in the following column vector.
\[ \small
    \begin{matrix}
        x^2 \\
        xy \\
        y^2 \\
        x \\
        y \\
        1
    \end{matrix}
    \begin{bmatrix}
        5\\7\\2\\0\\1\\6
    \end{bmatrix}
\] \normalsize
The coefficient vector of the polynomial $p(x,y)$ is defined as the column vector of length\footnote{Indeed, for each $i$, there are $i+1$ ways to choose non-negative integer exponents $a$ and $b$ such that $a+b=i$, and since the exponents can sum to any integer between $0$ and $d$, there are $\sum_{i=0}^d (i+1) = \frac{(d+1)(d+2)}{2}$ choices of exponents (and thus this many possible coefficients).} $\frac{(d+1)(d+2)}{2}$ whose entries are the coefficients of $p$, ordered by the decreasing degree and within each degree in the order of the decreasing powers of $x$. More precisely, the entry of $v$ that corresponds to $x^ay^b$ at index $(a,b)$ where $a+b=d$, is $c_{ab}$.

Let $\Delta_d$ be the vector of length $d+1$ with $i$-th entry as ${d \choose i-1}$.
Let $A_d$ be the block matrix with $(d+1)$ rows and $(d+1)(k+1)$ columns defined as  follows: $A_d = \left[ A^{(d)} \mid A^{(d-1)} \mid \cdots \mid A^{(1)} \mid A^{(0)} \right]$ - where $A^{(k)}$ is a matrix of order $(d+1) \times (k+1)$ defined as follows : For $1 \le i \le d+1$, and $1 \le j \le k+1$, the $(i,j)$-th entry of $A^{(k)}$ is defined as: $a_{ij}^k = {d-k \choose i-j}$ if $j \le i \le j+d-k$, and $0$ otherwise.

\begin{theorem}[\cite{WG14}]
Let $p(x,y)$ be a bivariate polynomial of degree $d$ with non-negative coefficients, with a coefficient vector $v$. Let $A_d$ be the corresponding matrix defined above. Then, $p$ is a tree polynomial if and only if $\left( A_d \right) v = \Delta_d$ and
$$c_{ab} \le {a+b \choose a} - \sum_{\substack{0 \le i+j \le a+b \\ i \le  a \\ j \le b}} c_{ij}{a+b-i-j \choose a-i}$$
\end{theorem}
Thus, the polynomial is a tree polynomial if the corresponding coefficient vector is in an affine space defined by $Av = b$ with certain restrictions on the entries of the vector.

\begin{figure}[!h]
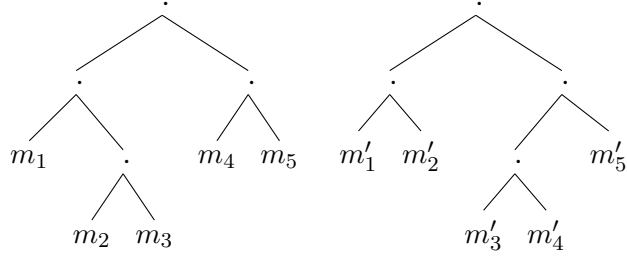

\vspace*{-2mm}
    \centering
   \scalebox{0.9}{ \Tree [.~. [.~. [.$m_1$ ] [.~. [.$m_2$ ] [.$m_3$ ] ] ] [.~. [.$m_4$ ] [.$m_5$ ] ] ]~~~\Tree [.~. [.~. [.$m_1'$ ] [.$m_2'$ ] ] [.~. [.~. [.$m_3'$ ] [.$m_4'$ ] ] [.$m_5'$ ] ] ] }
    \caption{Two distinct  full binary trees with same tree polynomial. Here $m_1=m_1', m_2=m_3', m_3=m_4', m_4=m_2', m_5=m_5' $}\vspace*{-2mm}
    \label{fig:PolygonNonUniqueness}
\end{figure}

\eject

A natural approach is to use the connection between tree polynomials to explore the rotation distance problem. A first question is to ask if the polynomial uniquely describes the tree. Unfortunately, it is not the case as we describe below and depict in Figure \ref{fig:PolygonNonUniqueness}.

\begin{proposition}
For any $n \ge 4$, there exists two distinct  full binary trees $T_1$ and $T_2$ with $n$ internal nodes, such that $p_{T_1}$ and $p_{T_2}$ are identical.
\end{proposition}
\begin{proof}
For any $n \ge 4$, there are trees for which at least two leaves representing the same monomial. For such a tree $T$, let $p_T = m_1 + m_2 + \ldots + m_n$ be the polynomial representing a tree with $n$ leaves. Let $m_i = m_j$ be a repeating monomial. We now construct two different trees $T_1,T_2$ by replacing the leaf node represented by $m_i,m_j$ respectively by an internal node with two children. Note that Any leaf represented by monomial $m$, after replacement with a node generates two monomials $xm$ and $ym$. Thus, $p_{T_1} = m_1+m_2+\ldots xm_i+ym_i+\ldots m_n$ and  $p_{T_2} = m_1+m_2+\ldots xm_j+ym_j+\ldots m_n$. $T_1=T_2$ as $m_i=m_j$. Clearly, $p_{T_1} \ne p_{T_2}$, but $T_1$ and $T_2$ are distinct trees as we constructed them by altering two different leaf nodes of the same tree.
\end{proof}

Note that the above example in Figure \ref{fig:PolygonNonUniqueness} that we have constructed is of rank $2$ and hence we cannot use the polynomial representation for rotation distance even when restricted to rank $2$. However, the situation is better when the rank of the tree is restricted to $1$, which, as we described earlier, is corresponds exactly to skew trees. We study this below.

\paragraph{Skew Polynomials:}
A bivariate polynomial $p(x,y)$ is said to be a {\em skew polynomial} if there is a full skew binary tree whose tree polynomial is $p(x,y)$. A first question is to characterize them.

\smallskip
The following necessary conditions are easy to observe: if $d$ is the height of the tree, there are exactly two monomials of degree $d$ in the skew polynomial and for every other degree, there is exactly one monomial. But, this is not sufficient as observed by the following example $x+x^2+y^2$ which cannot be a polynomial for any full skew tree. If we put an additional condition, then this indeed gives a characterization as we show below.

\begin{proposition}[\textbf{Characterization of Skew Polynomials}]
A bivariate polynomial $p(x,y)$ represents a full skew tree with $n$ internal nodes if and only if
\begin{enumerate}[label=(\arabic*)]
\itemsep 0pt
\item the degree of the polynomial is $n$ with $n+1$ terms, in which there is a unique monomial of degree $i$ for every $i \in [n-1]$, and there exists two monomials of degree $n$.
\item $\forall i, 1< i < n$, if $m_i$ and $m_{i-1}$ the monomial of degree $i$ and $i-1$ respectively, then $gcd(m_i, m_{i-1})$ has degree $i-1$ or $i-2$.
\item for $m_n, m_{n+1}$ of degree $n$ either gcd($\frac{m_n}{x}, \frac{m_{n+1}}{y}$) or gcd($\frac{m_n}{y}, \frac{m_{n+1}}{x}$) has degree $n-1$.
\end{enumerate}
\end{proposition}
\begin{proof}
We prove the reverse direction of the claim first (i.e. assume  a polynomial $p$ satisfying {\it (1)}, {\it(2)} and {\it(3)}.
By definition we know that at least one child of an internal node of a full skew tree is a leaf node. This will ensure that the polynomial corresponding to the tree consists of terms of all degrees from $1$ to $n$. In addition, we know that there exists two leaf nodes at the bottom most level. This create one extra term of degree $n$, thus producing a polynomial with $n+1$ terms.

\medskip
We now discuss about the relation between two terms having a degree difference of $1$ among them. We know that the grandparent of a leaf node at height $i$ and the parent of the leaf node at height $i-1$ are same. This implies that the two monomials representing these have at least $i-2$ terms in common. There are two possible ways in which leaves at the consecutive levels can appear.
\begin{description}
\item{\bf Case 1: $i^{th}$ leaf is right (resp. left) and $(i-1)^{th}$ leaf is left (resp. right)}
The term representing $(i-1)^{th}$ leaf will have one extra term  $x$ (resp $y$) in the term and the one representing $i^{th} $ leaf will have two extra $y$ (resp. $x$) with respect to the term that represents its first common ancestor node if it is a leaf. Thus here the gcd of the two terms will have a degree difference of $2$ comparing with the $i^{th} $ term.

\item{\bf Case 2: $i^{th}$ as well as $(i-1)^{th}$ leaves are either both left or both right}
If both leaves are on the left (resp. right), both will have the extra factor of $x$ (resp. $y$) in the term comparing with the first common ancestor node. The difference in degree for $i^{th}$ leaf is achieved by one extra $y$(resp. $x$) with respect to the $(i-1)^{th} ~$ term. This is the cause for a degree difference of just $1$ between the $i^{th}$ term and $(i-1)^{th} $term.
\end{description}

Suppose the $i^{th}$ degree term is $x^ay^b$. It can be  noted that, if the degree of the highest common factor is $i-2$, then the common factor is either $x^{a}y^{b-2}$ or $x^{a-2}y^{b}$ and not $x^{a-1}y^{b-1}$.
The two monomials $m_n, m_{n+1}$ of degree $n$ each have the property that either gcd($\frac{m_n}{x}, \frac{m_{n+1}}{y}$) or gcd($\frac{m_n}{y}, \frac{m_{n+1}}{x}$) has degree $n-1$. These two correspond to the leaves of the internal node at the last level $n$.

\medskip
Now we prove the forward direction of the claim. We show how to construct a full skew tree to a polynomial having  the above mentioned property. We know that the polynomial is of degree $n$ with $n+1$ terms. We also know that there exists exactly one term of all degrees from $1$ to $n-1$. The term having degree $1$ is either $x$ or $y$. If the degree $1$ term is $x$ (resp. $y$), then we construct a tree whose left (resp. right) child of root node is a leaf. For any term of the form $x^a y^b$, the possible terms with one degree less is one among $x^ay^{b-1}, x^{a-1}y^b, x^{a+1}y^{b-2}, x^{a-2}y^{b+1}$. If the gcd of $m_i$ and $m_{i-1}$ has degree $1$, then the leaf at height $i$ is on the same side as that on $(i-1)^{th}$. If the gcd of these terms has degree $2$, then the leaf at height $i$ is on the opposite side with respect to the leaf at height $(i-1)$. This way the skew tree that represent the polynomial is constructed. As there are two terms with degree $n$ present in the polynomial, those two corresponds to the two leaves of the node at height $n$. This way we create the unique skew tree from the given polynomial with the above mentioned properties.
\end{proof}

\paragraph{Algorithm for Skew Rotation Distance using Skew Polynomials:}
We notice now that the previous algorithm can alternatively be viewed in terms of polynomials as well. We call a non-root internal node (other than internal node $n$)in a skew tree an {\em angle node} if the leaf child is on the left (or resp. right) and the leaf child of its parent is on the right (or resp. left). Given a skew polynomial, a term $x^ay^b$ represents the leaf of an angle node implies that the polynomial should contain either the term $x^{a+1}y^{b-2}$ or $x^{a-2}y^{b+1}$. Internal node at level $n-1$ is always considered as an angle node since rotation is possible with respect to it with the skewness property maintained. The internal node at level $n$ is not an angle node as rotation creates a non-full skew tree.
We know that the skew polynomial contains two terms of degree $n$. If $p$ represents such a polynomial, $p_n$ and $p_{n+1}$ represents these two. The parent of these two monomials is the node at level $n$. We propose an algorithm to find the minimum skew rotation distance between two full skew trees and prove it.

\begin{theorem}
{\sc SkewPolynomialRotationDistance} (Algorithm~\ref{algo:polyskewrot}) algorithm outputs the minimum skew rotation distance between two skew trees given their corresponding skew polynomials.
\end{theorem}

\begin{algorithm}[h]
\KwIn{Two skew polynomials $p$ and $q$ of length $n+1$, representing the trees $T_1$ and $T_2$.}
\KwOut{Skew rotation distance between $T_1$ and $T_2$}

{\tt count = 0}

    \While{$p_1\neq q_1$}{
    Find smallest $i$ where $p_i$ or $q_i$ represents an angle term.

        \eIf{$p_i$ {\tt represents an angle node}}{
            {\tt doRotation(p,i)}

            {\tt count = count + 1}

        }
        {
            {\tt doRotation(q,i)}

            {\tt count = count + 1}
        }

    }

\While{$p \neq q$}{
    {\tt i = 1}

    \While{i $<$ n} {
    \eIf{$p_i \neq q_i$}{
        \eIf{$p_i$ {\tt represents an angle node}}
        {
            \tt doRotation(p,i)

            {\tt count = count+1}

            break
        }
        {
            \tt doRotation(q,i)

            {\tt count = count+1}

            break
        }

   }
   {
        i = i+1
   }
   }
}
\Return{\tt count}\;
\caption{{\sc SkewPolynomialRotationDistance}}
\label{algo:polyskewrot}
\end{algorithm}

\begin{proof}
Given two full skew trees $T_1, T_2$ on $n$ internal nodes represented by  skew tree polynomials $p$, $q$ respectively, we use Algorithm~\ref{algo:polyskewrot} to find the minimum skew rotation distance between them. The algorithm runs until the input skew tree polynomials $p$ and $q$ matches. Here, both the polynomials gets altered to an intermediate full skew tree that lie on a shortest skewness maintained path. The procedure $doRotation(p,i)$  (Algorithm~\ref{algo:rotation@i}) is used to alter the skew tree polynomial $p$ to the one that correspond to the one after rotation at internal node $i$.

\begin{algorithm}[h]
\KwIn{A skew polynomial $p= \sum_{j=1}^{n+1} m_j $, representing $T$ and an index $i$ the point of rotation}
\KwOut{Skew polynomial corresponding after a single rotation}
      \If{$(m_i/m_{i-1})==x^{-1}y^2$}{
        $temp = m_i$ \\
        $m_i = m_{i+1}/x$\\
        $m_{i+1} = temp*y$\\
       }
       \If{$(m_i/m_{i-1})==x^2y^{-1}$}{
        $temp = m_i$ \\
        $m_i = m_{i+1}/ x$\\
        $m_{i+1} = temp*x$
       }
\Return{}\;
\caption{{\tt doRotation(p,i)}}
\label{algo:rotation@i}
\end{algorithm}

The objective is to find the minimum skew rotation distance from the corresponding polynomials representing full skew trees $T_1$ and $T_2$. Our algorithm works by comparing nodes in increasing height. This is achieved by considering terms in  ascending order of its degree. As a first step, the algorithm  makes the first term identical. This is important as it may be the case that the input polynomials are representing right comb and left comb. In this case a rotation at root node may not generate a skew tree, thus violating the skewness property.  But, we know that rotation between the last two nodes in a skew tree does not change the skewness property. Thus it is always the case that our algorithm  finds an angle node.

By convention the second last node is considered as an angle node. We now argue that the algorithm terminates and outputs a count.

\noindent {\bf Termination:} If both skew trees are different, then  the polynomials that represents them are also different. So there exists two monomials of same degree which are different. Our algorithm  finds the nearest angle node nearest from the root. That is, it identifies the lowest degree monomials which are different. Then it changes the polynomial that correspond to a rotation, thereby matching the corresponding monomials. The algorithm terminates when both the polynomials become identical and outputs the number of rotations made.

\noindent {\bf Correctness:} Algorithm \ref{algo:polyskewrot} more or less imitates
Algorithm \ref{algo:binskewrot}. The difference is on the representation of full skew binary trees. A rotation at angle node ensures that the skewness property is maintained. Since polynomials generated after each rotation correspond to a skew tree and as in Algorithm \ref{algo:binskewrot}, rotation is done on the nearest angle node from the root where both trees differ, minimum number of rotations are done. Thus the algorithm outputs the minimum skew rotation distance.
\end{proof}

\section{Discussion and open problems}
\label{sec:openproblems}
In this paper, we studied algorithmic aspects of rotation distance using rank as an underlying parameter. A degenerate case turned out to be the case when the rank is $1$, and in this case we provided a polynomial time algorithm which determines the shortest rotation sequence transforming one skew tree to another, ensuring that all the intermediate trees are also skew. We list down the following points related to the exploration which also leads to open problems.

\begin{description}
\item{\bf $\NP$-hardness of Rotation Distance Problem:}  We observed that the operation of rotation can be interpreted as special transpositions on permutations. Noticing that sorting by transpositions is an $\NP$-hard problem, this connection may help in establishing the $\NP$-hardness of the problem.
\item{\bf Reducing the problem to special cases:} We noted that there is a reduction from the general rotation distance problem to the case where the given trees are of a particular rank $r$ and we have to find the shortest path which preserves the rank of the intermediate trees as well to be exactly $r$. Similarly, we noted that there is a reduction from the general rotation distance problem to the case where the given trees are of height $h$ and we have to find the shortest path which preserves the height of the intermediate trees as well to be exactly $h$. It is an interesting open problem to see if this additional information/restriction about rank or height can be used to design algorithms.
\end{description}

\subsection*{Acknowledgement:} We thank the anonymous reviewers for the extensive review comments which improved the readability and for pointing out gaps in an earlier version of the proof of Lemma~\ref{lem:n2-rotations-to-reduce-rank}.

\end{document}